\documentclass[conference]{IEEEtran}
\usepackage[utf8]{inputenc}
\usepackage{amsmath, amssymb, amsthm}
\usepackage{enumerate}
\usepackage{graphicx}
\usepackage{xcolor}
\usepackage{booktabs}

\usepackage{algorithm, algorithmicx, algpseudocode}
\algrenewcommand\algorithmicrequire{\textbf{Input:}}
\algrenewcommand\algorithmicensure{\textbf{Output:}}
\algnewcommand{\BlackBox}[1]{%
    \begin{flushleft}
    \hspace{-.7cm}
    \textbf{Available Functions:}
    {\raggedright #1}
    \end{flushleft}
}
\algnewcommand{\Initialize}[1]{%
    \begin{flushleft}
    \hspace{-.7cm}
    \textbf{Initialize:}
    {\raggedright #1}
    \end{flushleft}
}

\usepackage{tikz}
\usetikzlibrary{shapes.geometric, arrows}
\tikzstyle{startstop} = [rectangle, rounded corners, 
minimum width=2cm, 
minimum height=1cm,
text centered, 
draw=black, 
fill=gray!10]
\tikzstyle{arrow} = [thick,->,>=stealth]

\usepackage{hyperref}
\usepackage[capitalize,nameinlink]{cleveref}
\usepackage{centernot}

\usepackage{color, soul}

\usepackage{mdframed}
\mdfdefinestyle{MyFrame}{%
    linecolor=black,
    outerlinewidth=2pt,
    roundcorner=20pt,
    innertopmargin=\baselineskip,
    innerbottommargin=\baselineskip,
    innerrightmargin=10pt,
    innerleftmargin=10pt,
    backgroundcolor=gray!20!white
}

\newcommand{\CC}{{\mathcal C}}

\newcommand{\CF}{{\mathcal F}}
\newcommand{\CH}{{\mathcal H}}

\newcommand{\CP}{{\mathcal P}}

\newcommand{\ol}{\overline{L}}
\newcommand{\fpkq}{\CF_{p}(k, q)}

\newcommand{\BC}{{\mathbb C}}

\newcommand{\BF}{{\mathbb F}}

\newcommand{\BR}{{\mathbb R}}
\newcommand{\BZ}{{\mathbb Z}}

\newcommand{\fq}{\BF_{q}}

\newcommand{\fqx}{\BF_{q}[X]}
\newcommand{\fqxy}{\BF_{q}[X, Y]}

\newcommand{\dH}{d_{{\rm H}}}
\newcommand{\ds}{d^{({\rm s})}}
\newcommand{\dc}{d_{{\rm c}}}

\newcommand{\herm}[1]
{{#1}^{\dagger}}

\DeclareMathOperator{\wt}{wt}

\DeclareMathOperator{\CCP}{CP}
\DeclareMathOperator{\RS}{RS}
\DeclareMathOperator{\GRS}{GRS}

\DeclareMathOperator{\GS}{GS}
\DeclareMathOperator{\RSdecode}{GRS-decode}
\DeclareMathOperator{\CPdecode}{CP-decode}
\DeclareMathOperator{\CPlistdecode}{CP-list-decode}

\newtheorem{theorem}{Theorem}

\theoremstyle{definition}
\newtheorem{definition}{Definition}

\theoremstyle{remark}
\newtheorem{remark}{Remark}

\newcommand{\question}[1]
{{\color{red} {\bf Question:} #1}}

\newcommand{\answer}[1]
{{\color{blue} {\bf Answer:} #1}}

\newcommand{\note}[1]
{{\color{purple} {\bf Note:} #1}}

\newcommand{\todo}[1]
{{\color{red} {\bf To Do:} #1}}

\newcommand{\hide}[1]
{{\iffalse #1 \fi}}

\IEEEoverridecommandlockouts
\title{Decoding Analog Subspace Codes: Algorithms for Character-Polynomial Codes}
\author{\IEEEauthorblockN{Samin Riasat and Hessam Mahdavifar} 
\IEEEauthorblockA{Department of Electrical and Computer Engineering, Northeastern University, Boston, MA 02115, USA \\ 
Email: \{\href{mailto:riasat.s@northeastern.edu}{riasat.s}, \href{mailto:h.mahdavifar@northeastern.edu}{h.mahdavifar}\}@northeastern.edu}
\thanks{This work was supported by NSF under Grant CCF-2415440 and the Center for Ubiquitous Connectivity (CUbiC) under the JUMP 2.0 program.}

}
\date{}

\begin{document}

\maketitle

\begin{abstract}
    We propose efficient minimum-distance decoding and list-decoding algorithms for a certain class of analog subspace codes, referred to as character-polynomial (CP) codes, recently introduced by Soleymani and the second author. In particular, a CP code without its character can be viewed as a subcode of a Reed--Solomon (RS) code, where a certain subset of the coefficients of the message polynomial is set to zeros. We then demonstrate how classical decoding methods, including list decoders, for RS codes can be leveraged for decoding CP codes. For instance, it is shown that, in almost all cases, the list decoder behaves as a unique decoder. We also present a probabilistic analysis of the improvements in list decoding of CP codes when leveraging their certain structure as subcodes of RS codes. 
\end{abstract}

\section{Introduction}

Let $W$ be an ambient vector space and let $\CP(W)$ denote the set of all of subspaces of $W$. Then a subspace code associated with this ambient vector space is defined as a subset of $\CP(W)$. In general, subspace codes become relevant in non-coherent communication settings where the communication medium only preserves the subspace spanned by the input vectors into the medium rather than the individual entries of the vectors. The notion of subspace codes and their application in non-coherent communication was first developed in the seminal work by Koetter and Kschischang in the context of randomized network coding \cite{KK}. Specifically, \cite{KK} develops a coding theory framework for subspace coding over finite fields. Recently, this framework is extended to the field of real/complex numbers in \cite{Hessam22}, i.e., when $W = \BR^M$ or $W = \BC^M$. In particular, it is shown in \cite{Hessam22} that \textit{analog} subspace codes enable reliable communication over wireless networks in a non-coherent fashion, where the communication medium is modeled as an \textit{analog operator channel}. This was done by demonstrating that the subspace error- and erasure-correction capability of an analog subspace code is characterized by the minimum subspace distance, a variation of the chordal distance introduced in \cite{Hessam22}. A new algebraic construction for one-dimensional complex subspace codes, referred to as character-polynomial codes, is introduced in \cite{Hessam22} and later extended to higher dimensions in \cite{Hessam21}.

Generally speaking, subspace codes are used for reliable communication over analog operator channels in the same way block codes are conventionally used for reliable communication over point-to-point channels. Such codes can also be viewed as codes in ``Grassmann space'' (i.e. \emph{Grassmann codes}), provided that the dimensions of all the subspace codewords are equal. Grassmann codes have also found applications in the design of communication system in the context of space-time code design for multiple-input multiple-output (MIMO) channels \cite{zheng2002communication}. However, the problem of Grassmann coding in MIMO settings is often viewed as a modulation constellation design. Furthermore, even though asymptotic limits (as $n$ grows large, and the dimension of subspace codewords is fixed) on packing in Grassmann manifolds are known \cite{barg2002bounds}, there are only a few explicit constructions known. The problem of designing efficient decoders for analog subspace codes/Grassmann codes is even less studied/understood in the literature. 

In this paper, we study CP subspace codes from a decoding perspective. In particular, we estimate the minimum distance of one-dimensional CP codes over complex numbers and provide a minimum-distance decoding algorithm. We also present a list-decoding algorithm and show that it is efficient by analyzing its complexity. Finally, we demonstrate that the list decoder almost always returns a list of size one. This is done by a probabilistic analysis of the average list size against the well-known Guruswami--Sudan list decoder for Reed--Solomon codes~\cite{Guruswami06}, demonstrating the advantages of leveraging CP code structures when utilizing off-the-shelf RS decoders. 

The rest of this paper is organized as follows. In \autoref{sec:two} we provide some preliminaries and background on subspace codes. In \autoref{sec:decoding} we present decoding algorithms for CP codes. The capabilities of the list-decoding algorithm are analyzed in \autoref{sec:prob}. Finally, the paper is concluded in \autoref{sec:five}. 

\section{Preliminaries: Operator Channel \texorpdfstring{\\}{ } and Subspace Codes}
\label{sec:two}


\subsection{Notation Convention}

\begin{itemize}
    \item 
    We view the elements of $W$ as row vectors and denote by $\CP_{k}(W)$ the set of all $k$-dimensional subspaces of $W$. 
    \item Polynomial and power series coefficients are expressed using subscripts on the symbol of the polynomial or power series, e.g. $f(X) = \sum_{k} f_{k} X^{k}$. This is done similarly for vector and codeword coordinates, e.g. $v = (v_{1}, \dots, v_{n})$. 
\end{itemize}

\subsection{Analog Operator Channel}

\begin{definition}
[{\cite[Definition~1]{Hessam22}}]
    \label{def:channel}
An \emph{(analog) operator channel} $C$ associated with $W$ is a channel where the input subspace $U \in \CP(W)$ and the output subspace $V \in \CP(W)$ are related by
    \begin{align*}
        V = \CH_{k}(U) \oplus E,
    \end{align*}
    where
    \begin{itemize}
        \item $\CH_{k}$ is a stochastic operator that
        returns a random $k$-dimensional subspace of $U$. 
        In particular, $\CH_{k}(U) \in \CP_{k}(U)$ if $\dim(U) > k$, and $\CH_{k}(U) = U$ otherwise.
        \item $E \in \CP(W)$ is an \emph{interference/error subspace}, where, without loss of generality, $E \cap U = \{ 0 \}$.
    \end{itemize}
    In transforming $U$ to $V$, we say $C$ \emph{introduces $t = \dim(E)$ insertions/errors} and $\rho = \dim(V) - k$ \emph{deletions/erasures}.
\end{definition}

\subsection{Grassmann Space and Subspace Codes}

\begin{definition}
[{\cite[\S~I]{Hessam21}}, {\cite[\S~II]{Hessam22}}]
    $\CP_{m}(\BC^{n})$ is denoted $G_{m, n}(\BC)$ and is called the \emph{Grassmann space}. The elements of $G_{m, n}(\BC)$ are called \emph{$m$-planes}. Any $m$-plane $U$ is equipped with the natural inner product $\langle u, v\rangle := u \herm{v}$ for $u, v \in U$. 
\end{definition}


\begin{definition}
    \label{def:chordal}
    Let $U, V \in G_{m, n}(\BC)$ be $m$-planes with respective orthonormal bases $\{u_{i}\}_{1 \le i \le m}$ and $\{v_{i}\}_{1 \le i \le m}$ (i.e. $u_{i}, v_{i}$ are unit row vectors). The \emph{principal angle $\theta_{i}$} for $1 \le i \le m$ is $\theta_{i} := \arccos|u_{i} \herm{v_{i}}|$. Then the \emph{chordal distance} \cite{Conway96, Hessam22} between $U$ and $V$ is 
    \begin{align}
        \label{eq:8}
        \dc(U, V) 
        := \sqrt{\sum_{i = 1}^{m} \sin^{2} \theta_{i}}. 
    \end{align}
\end{definition}

We remark that this is not the only possible way to define distances between subspaces. However, as was shown in \cite{Hessam22}, a variation of this notion can perfectly capture the error- and erasure-correction capabilities of subspace codes for the analog operator channel, as defined below.

\hide{
\begin{definition}
    \label{def:quasi}
    Let $\sigma > 1$. A \emph{$\sigma$-quasimetric} on a set $M$ is a function $d: M \times M \to \BR^{+}$ that satisfies the conditions of a metric except that the triangle inequality is relaxed to the \emph{$\sigma$-relaxed triangle inequality}:
    \begin{align*}
        d(x, z) < \sigma[d(x, y) + d(y, z)] 
    \end{align*}
\end{definition}
}

\begin{definition}[{\cite[Definition~2]{Hessam22}}]
    \label{def:dist2}
    The \emph{subspace distance} of two subspaces $U, V \in \CP(W)$ 
    is defined as 
    \begin{align}
        \label{eq:12}
        \ds(U, V) 
        &:= 2 d_c(U,V)^2. 
    \end{align}
    \hide{
    Observe that this is a \emph{$2$-quasimetric}. 
    }
\end{definition}

\hide{
\begin{remark}
    If $U, V \in G_{1, n}(\BC)$, then 
    \begin{align*}
        \dc(U, V) 
        &= |\sin^{2}(\arccos|u_{1} v_{1}^{\dagger}|)| 
        = 1 - (u_{1} v_{1}^{\dagger})^{2} 
    \end{align*}
    Hence, 
    \begin{align*}
        d(U, V) 
        &= 2 (1 - (u_{1} v_{1}^{\dagger})^{2})^{2} 
    \end{align*}
    \todo{
        How does this play with the Hamming metric used for $\CCP$ later on? Does this generalise the Hamming metric?
    }
\end{remark}
}

\begin{definition}
[{\cite[Definition~3]{Hessam22}}]
    A \emph{subspace code} $\CC$ is a collection of subspaces of $W$, 
    i.e. $\CC \subseteq \CP(W)$.
    The \emph{minimum distance} of $\CC$ is 
    \begin{align*}
        \ds_{\min}(\CC) 
        &:= \min_{\substack{U, V \in \CC \\ U \neq V}} \ds(U, V). 
    \end{align*}
\end{definition}


\subsection{Minimum-Distance Decoding}

As with conventional block codes, one can associate a minimum-distance decoder to a subspace code $\CC$ for communication over an analog operator channel in order to recover from subspace errors and erasures. Such a decoder returns the nearest codeword $V \in \CC$ given $U \in \CP(W)$ as its input, i.e. for any $V' \in \CC$, $\ds(U, V) \le \ds(U, V')$.

\begin{theorem}[{\cite[Theorem~1]{Hessam22}}]
    \label{thm:1}
    Consider a subspace code $\CC$ used for communication over an analog operator channel as in \autoref{def:channel}. Then the minimum-distance decoder successfully recovers the transmitted codeword $U \in \CC$ from the received subspace $V$ if 
    \begin{align*}
        2(\rho + t) < \ds_{\min}(\CC). 
    \end{align*}
\end{theorem}

\autoref{thm:1} implies that erasures and errors have equal costs in the subspace domain as far as the minimum-distance decoder is concerned. In other words, the minimum-distance decoder for a subspace code $\CC$ can correct up to $\left\lfloor \frac{\ds_{\min}(\CC) - 1}{2} \right\rfloor$ errors and erasures.

\subsection{Construction of Analog Subspace Codes}


We recall below the novel approach \cite{Hessam22} based on character sums resulting in explicit constructions with better rate-minimum distance trade-off compared to the other known constructions for a wide range of parameters. 
We also introduce some new notation that simplifies our presentation and arguments compared to earlier work. 
Henceforth, $\fq$ denotes a finite field of characteristic $p$ with $q$ elements. 

\begin{definition}
    \label{def:fpkq}
    For $k < q$, the \emph{message space} 
    \begin{align*}
        \CF(k, q) 
        &:= \{f \in \fqx: \deg(f) \le k\} 
    \end{align*}
    is the set of all polynomials of degree at most $k$ over $\fq$. The elements of $\CF(k, q)$ are called \emph{message polynomials}, whose coefficients represent message symbols. 
    For 
    $f \in \CF(k, q)$, we denote 
        $f^{(p)}(X) 
        := \sum_{0 \le j \le k / p} 
        f_{j p} X^{j}$
    and define the additional message spaces 
    \begin{align*} 
        \fpkq 
        &:= \{f(X) - f^{(p)}(X^{p}): f \in \CF(k, q)\} \\ 
        \fpkq' 
        &:= \{f(X) / X: f \in \fpkq\}.
    \end{align*}
    In other words, $\fpkq$ is the set of all message polynomials  $f \in \CF(k, q)$ with $f_{j p} = 0$ for all integers $j \ge 0$. 
\end{definition}

\begin{definition}[CP Code {\cite[Definition~6]{Hessam21}}]
    \label{def:cp}
    Fix $k \le n = q - 1$, a non-trivial character $\chi$ of $\fq$, and any ordering $\alpha_{1}, \dots, \alpha_{n}$ of $\fq^{\times} := \fq \setminus \{0\}$
    . Then the encoding of $f \in \fpkq$ in $\CCP := \CCP(\fpkq, \chi) \subseteq G_{1, n}(\BC)$ is given by
    \begin{align*}
        \CCP(f) 
        &:= (\chi(f(\alpha_{1})), \dots, \chi(f(\alpha_{n}))). 
    \end{align*}
\end{definition}

We remark that all choices for the non-trivial character $\chi$ result in the same codebook~\cite{Hessam22}. 
Hence, we will sometimes omit one or more of the parameters $k, p, q, \chi$ when they are clear from the context. Note 
also that 
\begin{align}
    \label{eq:CPsize}
    |\CCP| = |\fpkq| 
    = |\fpkq'| = q^{k - \lfloor k / p \rfloor},
\end{align}
which can be obtained from \autoref{def:fpkq}. (See also \cite[Theorem~9]{Hessam22}.) 

\section{Decoding CP Codes}
\label{sec:decoding}

\hide{
\textcolor{red}{One can observe the image of a CP code in $\fq$ can be regarded as a subcode of a Reed-Solomon code, where a certain subset of the coefficients of the message polynomial is set to zero.} 
}
One can observe that $\CCP(\fpkq, \chi)$ without $\chi$ applied to the codeword coordinates can be regarded as a subcode of a Reed--Solomon (RS) code, where a certain subset of the coefficients of the message polynomial is set to zeros. 
Hence, RS decoders may be used as one component in decoding CP codes. We recall the following definitions of the well-known RS codes that will be referred to throughout this work. 

\begin{definition}[RS Code]
    \label{def:rs}
    Fix $k \le n \le q$ and distinct $\alpha_{1}, \dots, \alpha_{n} \in \fq$. Then the encoding of $f \in \CF(k - 1, q)$ 
    in $\RS := \RS(\CF(k - 1, q))$ is given by
    \begin{align}
        \label{eq:rs}
        \RS(f) 
        &:= (f(\alpha_{1}), \dots, f(\alpha_{n})). 
    \end{align}
\end{definition}

\begin{definition}[Generalized RS Code]
    \label{def:grs}
    Fix $k \le n \le q$, distinct $\alpha_{1}, \dots, \alpha_{n} \in \fq$ and not necessarily distinct 
    $v_{1}, \dots, v_{n} \in \fq^{\times}$. Then the encoding of $f \in \CF(k - 1, q)$ in $\GRS := \GRS(\CF(k - 1, q))$ is given by
    \begin{align}
        \label{eq:grs}
        \GRS(f) 
        &:= (v_{1} f(\alpha_{1}), \dots, v_{n} f(\alpha_{n})).
    \end{align}
\end{definition}

For the rest of this section, let $n = q - 1$, $d = n - k + 1 \ge 1$ and $v_{i} = \alpha_{i}$ for $i \in \{1, \dots, n\}$. 
Note that $n = q - 1$ is a common special case for $\RS$ and $\GRS$
,
and it is well known that $\RS$ in \autoref{def:rs} and $\GRS$ in \autoref{def:grs} are $[n, k, d]_{q}$ codes. 

\hide{
\begin{definition}
    \label{def:subcode}
    Given $\RS$ and $\GRS$ as above, for any $S \subseteq \CF(k, q)$, define 
    \begin{align*}
        \RS(S) 
        &:= \{\RS(f): f \in S\} \\ 
        \GRS(S) 
        &:= \{\GRS(f): f \in S\}
    \end{align*}
\end{definition}
}

Note that as in \autoref{def:cp}, there was a slight abuse of notation by identifying the $\RS$ and $\GRS$ codes above with their encoding maps \eqref{eq:rs} and \eqref{eq:grs}. By doing so,
$\CCP(\fpkq, \chi)$ without $\chi$ applied to the codeword coordinates is simply expressed as $\RS(\fpkq) = \GRS(\fpkq')$ (i.e. see \eqref{eq:cp-grs} below), which is 
a linear subcode of $\GRS 
\subseteq \RS(\CF(k, q))$. 

Throughout the following, let $\CC = \GRS(\fpkq')$. 


\subsection{Unique Decoding}

\subsubsection{Minimum Distance 
of CP Codes}

Since the minimum distance of $\RS$ and $\GRS$ as defined above is $d$,
a $\RS$ or $\GRS$ decoder can correct $< d / 2$ block errors. 
Consequently, the decoder applied to 
$\CC$ 
can also correct $< d / 2$ block errors. 
It is therefore natural to ask whether polynomials in $\fpkq'$ having $1 + \lfloor k / p \rfloor$ guaranteed zero coefficients can allow for correcting more errors. 

Computational evidence suggests that 
$\CC$
\begin{itemize}
    \item is MDS for $p = 2$ or $q = p$, 
    \item has minimum distance $
    d + \lfloor k / p \rfloor$ or 
    $d - 1 + \lfloor k / p \rfloor$
    . 
\end{itemize}
We try to prove some of these observations below. 

\begin{theorem}
    \label{thm:GRSdmin}
    The minimum distance of $\CC$ 
    satisfies 
    \begin{align*}
        d &\le 
        d_{\min}(\CC) 
        \le d + \left \lfloor \frac{k}{p} \right \rfloor. 
    \end{align*}
    with equality on the left if $k < p$ or $k - 1$ divides $q - 1$, and on the right iff $\CC$ is MDS. 
\end{theorem}

\begin{proof}
    The right half of the inequality follows directly from the Singleton bound, so we prove the other half.
    
    \hide{
    Observe that $\CC$ is a linear code, and 
    \begin{align*}
        \CCP 
        = \{(\chi(c_{1}), \dots, \chi(c_{n})): c \in \CC\}
    \end{align*}
    so that $d_{\min}(\CCP) = d_{\min}(\CC)$.
    }Since $\CC \subseteq \GRS$, 
    \hide{
    \begin{align*}
        \CC 
        := \RS(\fpkq) 
        &= \GRS(\fpkq') 
        \subseteq \GRS(\CF(k - 1, q)) 
    \end{align*}
    }
    \begin{align}
        \label{eq:grs-dmin}
        d_{\min}(\CC) 
        \ge d_{\min}(\GRS) 
        = d. 
    \end{align}
    In the other direction, let $f(X) = X^{k - 1} - 1 \in \fpkq'$. 
    When $k - 1$ divides $q - 1$, the $k - 1$ distinct roots of $f(X)$ all lie in $\fq^{\times}$. 
    Hence, $k - 1$ of the coordinate values in $\CC(f)$ are $0$, 
    which implies 
    $d_{\min}(\CC) \le \wt(\CC(f)) = d$.
    
    Finally, if $k < p$, then $\CC = \GRS$, so \eqref{eq:grs-dmin} is an equality. 
\end{proof}

\hide{
\subsubsection*{Computation Results}
\label{sec:dmin}

We computed the minimum distance of $\CC = \GRS(\fpkq')$ for $q = 16$ and $1 \le k < q$. Cases where $d_{\min}(\CC) > d$ are shown in \autoref{tab:dmin}. The values of $d_{\min}(\CC)$ achieving equality in the Singleton bound are highlighted in red. 

\begin{table}[!htbp]
    \centering
    \caption{Minimum Distance of $\CC = \GRS(\fpkq')$}
    \begin{tabular}{|c|c|c|}
        \hline 
        $(q, k)$ & $d_{\min}(\CC)$ & $d$ \\
        \hline 
        $(16, 3)$ & ${\color{red} 14}$ & $13$ \\ 
        $(16, 5)$ & ${\color{red} 13}$ & $11$ \\ 
        $(16, 7)$ & ${\color{red} 12}$ & $9$ \\ 
        $(16, 9)$ & ${\color{red} 11}$ & $7$ \\ 
        $(16, 10)$ & $8$ & $6$ \\ 
        $(16, 11)$ & ${\color{red} 10}$ & $5$ \\ 
        $(16, 12)$ & $8$ & $4$ \\ 
        $(16, 13)$ & $9$ & $3$ \\ 
        $(16, 14)$ & $8$ & $2$ \\ 
        $(16, 15)$ & $8$ & $1$ \\ 
        \hline 
    \end{tabular}
    \label{tab:dmin}
\end{table}
}


Note that one-dimensional CP subspace codes can be regarded as classical block codes by ensuring that all codewords have equal norms. Hence, their minimum Hamming distance $d_{\min}(\CCP)$ can be also studied when regarded as block codes, as in the next theorem. 

\begin{theorem}
    \label{thm:CPdmin} 
    The minimum distance of $\CCP$ 
    satisfies 
    \begin{align}
        \label{eq:CPmindist1}
        d_{\min}(\CCP) 
        &\le d 
        + \left \lfloor \frac{k}{p} \right \rfloor 
    \end{align}
    with equality iff $\CCP$ is MDS. Furthermore, if $q = p$, then
    \begin{align} 
        \label{eq:CPmindist2}
        d_{\min}(\CCP) 
        &\ge d. 
    \end{align}
    In particular, $\CCP$ is MDS when $k < p$ or $q = p$. 
\end{theorem}

\begin{proof}
    Observe that 
    \begin{align}
        \CCP(\fpkq, \chi) = 
        \{(\chi(c_{1}), \dots, \chi(c_{n})): c \in \CC\}.
        \label{eq:cp-grs}
    \end{align}
    Since $\chi(\fq)$ is the multiplicative group of the $p$-th roots of unity, it follows that
    $d_{\min}(\CCP) \le d_{\min}(\CC)$, with equality if $k < p$ or $q = p$. 
    Now \eqref{eq:CPmindist1} follows from the Singleton bound (via \eqref{eq:CPsize}), and \eqref{eq:CPmindist2} follows from \autoref{thm:GRSdmin}. 
\end{proof}

Note that \autoref{thm:GRSdmin} gives only sufficient conditions for equality in \eqref{eq:grs-dmin}. 
In other words, for other values of $k \ge p$ where $k-1$ does not divide $q - 1$, one may still find examples of $f \in \fpkq'$ with all roots in $\fq^{\times}$, which would also imply equality in \eqref{eq:grs-dmin}. 


\hide{
We now estimate the covering radius $\rho(\CC)$ of $\CC$. 

\begin{theorem}
    \label{thm:grs-rho}
    The covering radius of $\CC$ satisfies 
    \begin{align*}
        d - 1 
        &\le \rho(\CC) 
        \le d - 1 + \left \lfloor \frac{k}{p} \right \rfloor. 
    \end{align*}
\end{theorem}

\begin{proof}
    Since $\CC \subseteq \GRS$, 
    \begin{align*}
        \rho(\CC) 
        \ge \rho(\GRS) 
        = d - 1. 
    \end{align*}
    The other half follows from the redundancy bound 
    (e.g. see~\cite[Corollary~11.1.3]{Huffman03}). 
\end{proof}

\autoref{thm:grs-rho} gives the following result on the covering radius of $\CCP$ codes. 

\begin{theorem}
    \label{thm:CPrho} 
    If $q$ is prime, then the covering radius of $\CCP$ 
    is $d - 1$. 
\end{theorem}

\begin{proof}
    If $q$ is prime, then $q = p$ and 
    $k < p$. Hence, $\rho(\CC) = d - 1$ by \autoref{thm:grs-rho}. In this case $\chi$ is also one-to-one, so the covering radius of $\CCP$ equals $\rho(\CC)$. 
\end{proof}
}

\autoref{thm:CPdmin} implies that when $q$ is prime, $\RS$ decoders may be used as a major component in decoding CP codes 
without losing any error-correction capability of the code in the worst-case scenario. The idea, based on \eqref{eq:cp-grs}, is as follows. Given any received word, we first map each coordinate to its closest point in $\Psi := \chi(\fq)$ and apply $\chi^{-1}$. Finally, we invoke any $\RS$ decoder to recover the original message. This is explained below in more detail. 

Henceforth, we take $q = p$. 

\subsubsection{CP Decoding Algorithm}
\label{sec:CPdecoder}

Observe that 
\begin{align*}
    \CCP&(\fpkq, \chi) \\ 
    &= \{(\chi(\alpha_{1} c_{1}), \dots, \chi(\alpha_{n} c_{n})): c \in \RS(\fpkq')\} \\
    &\subseteq \{(\chi(\alpha_{1} c_{1}), \dots, \chi(\alpha_{n} c_{n})): c \in \RS(\CF(k - 1, q))\}.
\end{align*}
Since $\Psi$ is a finite cyclic subgroup of 
$\{z \in \BC: |z| = 1\}$ of order $q$, 
we can map each coordinate $y_{i} \in \BC$ of the received message $y$ to its closest point 
$e^{2 \pi i r / q}$ in $\Psi$, where $r$ is the closest integer to $q \arg(y_{i}) / (2 \pi)$.
More precisely, define $\phi: \BC \to \fq$ as 
\begin{align}
    \label{eq:chi-inverse}
    z &\mapsto \exp \left( \frac{2 \pi i}{q} \left \lfloor \frac{q \arg(z)}{2 \pi} + \frac{1}{2} \right \rfloor \right). 
\end{align}
This leads to \autoref{alg:cp} below.
\hide{
\begin{enumerate}
    \item For $i \in \{1, \dots, n\}$, compute $y_{i} := 
    \alpha_{i}^{-1} 
    \phi(m_{i}')$.
    \item \label{alg:wb:2} Apply any $\RS$ decoder to $y$. 
    \item \label{alg:wb:3} For any $f$ returned in Line~\ref{alg:wb:2}, return $X f(X)$. 
\end{enumerate}
\todo{Rewrite using algorithm environment}
Suppose that there are no errors. Then, given 
\begin{align*}
    \chi(f(\alpha_{i})) 
    &= \chi(f_{0} + f_{1} \alpha_{i} + \cdots + f_{k} \alpha_{i}^{k}) \\ 
    &= \chi(f_{0}) \chi(f_{1} \alpha_{i}) \cdots \chi(f_{k} \alpha_{i}^{k}) \\ 
    &= \prod_{j = 0}^{k} 
    \chi(\alpha_{i}^{j})^{f_{j}}.
\end{align*}
for $i = 1, \dots, n$, we have to find $f$. 
In other words, given the system of linear equations 
\begin{align*}
    \log \chi(f(\alpha_{i})) 
    &= \sum_{j = 0}^{k} 
    f_{j} \log \chi(\alpha_{i}^{j}) 
\end{align*}
for $i = 1, \dots, n$ 
over $\BC$, we need to solve for the unknowns $f_{j} \in \fq$. 
With errors, we can likewise solve the system 
\begin{align*}
    \log \phi(m_{i}') 
    &= \sum_{j = 0}^{k} 
    f_{j} \log \chi(\alpha_{i}^{j}). 
\end{align*}
for $i = 1, \dots, n$. 
This is a system of $n$ equations in $k - \lfloor k / p \rfloor$ unknowns over $\BC$, which can be expressed in matrix form as $y = A x$, where $y \in \BC^{n}$ and $A \in \BC^{k \times n}$ with
\begin{align*}
    y_{i} 
    &= \log \phi(m_{i}'), \\ 
    A_{i, j} 
    &= \log \chi(\alpha_{i}^{j}), \\ 
    x_{j} 
    &= f_{j}. 
\end{align*}
This leads to \autoref{alg:cp} below.
}

\begin{algorithm}[!htbp]
    \caption{$\CPdecode(m')$} 
    \label{alg:cp}
    \begin{algorithmic}[1]
        \Require{Received word $m' \in \BC^{n}$ with $<d/2$ errors}
        \Ensure{Unique $f$ 
        with $\dH(m',\CCP(f)) < d/2$}
        \BlackBox{$\RSdecode$: any $\GRS$ decoder}
        \For {$i = 1, \dots, n$}
        \State $y_{i} \gets \phi(m_{i}')$ 
        \EndFor
        \State $g \gets \RSdecode(y)$ 
        \State \Return $X g(X)$
    \end{algorithmic}
\end{algorithm}

\hide{
\usepackage{tikz}
\usetikzlibrary{shapes.geometric, arrows}

\tikzstyle{all} = [rectangle, rounded corners, 
minimum width=1cm, 
minimum height=1cm,
text centered, 
draw=black, 
align=center, 
fill=gray!30]

\tikzstyle{arrow} = [thick,->,>=stealth]

\begin{tikzpicture}[node distance=2cm]

\node (poly) [all] {Message Polynomial \\ $f \in \mathbb{F}_{q}[X]$};
\node (rs) [all, below of=poly] {Reed--Solomon Codeword \\ $(f(\alpha_{1}), \dots, f(\alpha_{n}))$};
\node (cp) [all, below of=rs] {CP Codeword \\ $(\chi(f(\alpha_{1})), \dots, \chi(f(\alpha_{n})))$};

\draw [arrow] (poly) -- node[anchor=east] 
{$\alpha_{1}, \dots, \alpha_{n} \in \mathbb{F}_{p}$} 
(rs);
\draw [arrow, align=center] (rs) -- node[anchor=east] 
{additive character \\ $\chi: \mathbb{F}_{p} \to \mathbb{C}$} 
(cp);

\end{tikzpicture}
Let $\alpha^{d} = 1$. 
Then 
\begin{align*}
    \chi(f(\alpha)) 
    &= \chi(f_{0}) \chi(f_{d}) \cdots \chi(f_{1} \alpha) \chi(f_{d + 1} \alpha) \cdots. 
\end{align*}
In general, for any $d \mid q - 1$, 
\begin{align*}
    \chi(f(\alpha^{(q - 1) / d})) 
    &= \chi(f_{0}) \chi(f_{1} \alpha) \chi(f_{2}) \cdots \chi(f_{k} \alpha^{k \bmod 2)}). 
\end{align*}

\begin{theorem}
    $\CPdecode$ is a maximum-likelihood decoder. 
\end{theorem}

\begin{proof}
    Using the notation of \autoref{alg:cp}, let $c = \CCP(f)$. We will show that maximizing $\Pr[m' | c]$ is equivalent to minimizing $\dH(m', c)$. 

    Since the error at each coordinate is independent,
    \begin{align}
        \label{eq:ml-prob}
        \Pr[m' | c] 
        &= \prod_{i = 1}^{n} \Pr[m'_{i} | c_{i}]. 
    \end{align}
    Since the errors are random, 
    \begin{align*}
        \Pr[m'_{i} = c_{i}] 
        &= \Pr[m'_{i} = \chi(f(\alpha_{i}))] 
        = \frac{1}{p}. 
    \end{align*}
    Therefore, \eqref{eq:ml-prob} gives
    \begin{align*}
        \Pr[m' | c] 
        &= \left(\frac{1}{p}\right)^{\dH(m', c)} \left(1 - \frac{1}{p}\right)^{n - \dH(m', c)} \\ 
        &= \left(1 - \frac{1}{p}\right)^{n} \left(\frac{1}{p - 1}\right)^{\dH(m', c)}. 
    \end{align*}
    Since $p$ is fixed, maximizing $\Pr[m' | c]$ is equivalent to minimizing $\dH(m', c)$. 
\end{proof}
}


\subsection{List Decoding}

\hide{Since the minimum distance decoder for $\CCP$ (\autoref{alg:cp}) in general is only as good as any $\RS$ decoder
, we next look at list decoding. 
Plainly the Guruswami--Sudan list decoder for $\GRS$ \cite{Guruswami06} works for $\CCP$, but we want to try and improve it for $\CCP$ by utilising the special structure of $\fpkq'$. 
(as 
there are now fewer unknowns, i.e. non-zero coefficients).
}

In list decoding, the decoder is relaxed by allowing to return a small list of codewords containing the original codeword, rather than a single codeword. 
This often provides better error-correction capabilities at the expense of sacrificing uniqueness. 
List decoding of RS codes has been extensively studied in the literature (see e.g.~\cite{Guruswami06}). 
In particular, the well-known Guruswami--Sudan (GS) algorithm \cite{Guruswami06}
for $\GRS(\CF(k - 1, q))$ is 
used as a key step in our CP list-decoding algorithm below.
Here, 
for any received vector $y \in \BC^{n}$ and an integer parameter $s > 0$ (called \emph{interpolation multiplicity}), $\GS(y, s)$ 
returns all $f \in \CF(k - 1, q)$ such that $\dH(y, \GRS(f)) \le n - \tau_{s}$, 
where $\tau_{s} := \lfloor c / s \rfloor + 1$ and $c := \lfloor \sqrt{(k - 1) n s (s + 1)} \rfloor$. 
There are two major steps in the GS algorithm: the \emph{interpolation} step and the \emph{factorization} step. 

\subsubsection{CP List-Decoding Algorithm}
\label{sec:CPlistdecoder}

As with unique decoding, one can use 
GS to list decode CP codes by following the template of \autoref{alg:cp}. 
However, the advantage for CP codes, as will be shown in \autoref{sec:prob}, is that one can expect a much smaller list size and error probability compared to a plain RS code of the same length and degree. The resulting algorithm is formally stated below as \autoref{alg:cplist}. 
\hide{
Let $m' \in \BC^{n}$ be the received message with possibly $\ge d / 2$ errors. Based on GS, we propose the following list-decoding algorithm for CP: 

\begin{enumerate}
    \item \label{alg:cp:y} For $i \in \{1, \dots, n\}$, compute $y_{i} := 
    \alpha_{i}^{-1} 
    \phi(m_{i}')$.
    \item \label{alg:cp:gs} Apply the Guruswami--Sudan decoder to $y$. 
    \item \label{alg:cp:cp} Return $\{X f(X) \in \CF_{p}(k, q): f \in L\}$, where $L$ is the list returned in Line~\ref{alg:cp:gs}. 
\end{enumerate}
\todo{Rewrite using algorithm environment}
}
\begin{algorithm}[!htbp]
    \caption{$\CPlistdecode(m', s)$} \label{alg:cplist}
    \begin{algorithmic}[1]
        \Require{Received word $m' \in \BC^{n}$ with possibly $\ge d / 2$ errors; interpolation multiplicity $s > 0$}
        \Ensure{All $f \in \fpkq$ with $\dH(m', \CCP(f)) \le n - \tau_{s}$}
        \For {$i = 1, \dots, n$}
        \label{alg:cp:y}
        \State $y_{i} \gets \alpha_{i}^{-1} \phi(m_{i}')$ 
        \EndFor
        \label{alg:cp:yend}
        \State $L \gets \GS(y, s)$ \label{alg:cp:gs} 
        \State \Return $\{X g(X) \in \CF_{p}(k, q): g \in L\}$ \label{alg:cp:cp}
    \end{algorithmic}
\end{algorithm}


\begin{theorem}
    \label{thm:gs-radius}
    $\CPlistdecode(m', s)$ produces all codewords within distance $n - \tau_{s}$ of $m'$. 
\end{theorem}

\begin{proof}
    This is a direct consequence of \cite[Theorem~4.8]{Guruswami06}.
\end{proof}

\subsubsection{Complexity Analysis}

(Lines~\ref{alg:cp:y}-\ref{alg:cp:yend}) By pre-computing $\chi(\fq)$ in a hash table (using $O(n)$ space), 
we can compute $\phi(m_{j}')$ using \eqref{eq:chi-inverse} for each $j$ in $O(1)$ time. 
Therefore, we can compute $y$ in $O(n)$ time, or $O(1)$ latency with parallelization. 

(Line~\ref{alg:cp:gs}) Solving the interpolation problem 
in GS takes $O(n^{3})$ time using the Feng--Tzeng algorithm \cite[\S~VIII]{McEliece03} or naive Gaussian elimination. This was improved by K\"{o}tter to $O(s^{4} n^{2})$ \cite[\S~VII]{McEliece03}, which is the most efficient known solution~
\cite[\S~4]{McEliece03-1}. 

The most efficient known solution to the factorization problem 
of GS is due to Gao and Shokrollahi \cite{Shokrollahi00} with a time complexity of $O(\ell^{3} k^{2})$, 
although the Roth--Ruckenstein algorithm \cite[\S~V]{Roth00} is quite competetive \cite[\S~4]{McEliece03-1} 
(see also \cite[\S~IX]{McEliece03}). 
Here $\ell$ is a design parameter, typically a small constant \cite{Roth00}, which is an upper bound on the size of the list of decoded codewords, which is bounded above by the degree of the interpolation polynomial in the second variable.
~\cite{Guruswami06}. 

In general, $\ell \le \lfloor c / (k - 1) \rfloor
= O(s \sqrt{n / k})$ (see also \cite[Theorem~4.8]{Guruswami06}), which gives $O(\ell^{3} k^{2}) = O(s^{3} n \sqrt{k n})$, yielding a $O(s^{4} n^{2})$ time complexity for Line~\ref{alg:cp:gs} 
of the algorithm.

\hide{
\question{What value should we choose for $s$?}

\note{Keep $s$ as a parameter.} 
} 

\section{Probabilistic Analysis of List Decoding}
\label{sec:prob}

In this section, we analyze the decoder error probability and the average list size (averaged over all error patterns of a given weight) of \autoref{alg:cplist}. 

\hide{
\question{What is the probability that $\GS$ returns a list of size $> 1$?} 

\answer{McEliece~\cite[Appendix~D]{McEliece03} showed that $\GS$ almost always returns a list of size $1$ for $\RS(\CF(k - 1, q))$.} 

\note{Do a similar analysis for $\RS(\fpkq)$ to possibly get an improved result for CP (maybe run simulations).} 
}

Let $\CC \subseteq \fq^{n}$ be a linear code. 
Following McEliece \cite[Appendix~D]{McEliece03}, for an arbitrary vector $u \in \fq^{n}$, define 
\begin{align*}
    B_{\CC}(u, t) 
    &:= \{c \in \CC \setminus \{0\}: d(c, u) \le t\}, \\ 
    D_{\CC}(w, t) 
    &:= \sum_{\wt(u) = w} |B_{\CC}(u, t)|,
\end{align*}
i.e., if the all-zero vector is the transmitted codeword and $u$ is received, then $B_{\CC}(u, t)$ is the set of non-zero codewords at distance $\le t$ (the \emph{decoding radius}) from $u$. 
By linearity, if an arbitrary codeword $c$ is the transmitted and $u$ is the error pattern, then $|B_{\CC}(u, t)|$ is the number of codewords $\neq c$ (the \emph{non-casual codewords}) at distance $\le t$ from the received word $c + u$. 
If $|B_{\CC}(u, t)| = s$, we say $u$ is \emph{$s$-tuply falsely decodable}. 
Then, $D_{\CC}(w, t)$ is the total number of falsely decodable words $u$ of weight $w$, where an $s$-tuply falsely decodable word is counted $s$ times. 

Recall that a bounded distance decoder with decoding radius $t$ returns all codewords within distance $t$ of the received word. 

\begin{theorem}
    \label{thm:CPlist}
    Let $\CC \subseteq \fq^{n}$ be a linear code with minimum distance $d$. 
    Consider a bounded distance decoder for $\CC$ with decoding radius $t$. If the error vector has weight $w$, then the average number of non-casual codewords (over all error patterns of weight $w$) in the decoding sphere is given by
    \begin{align*}
        \ol_{\CC}(w, t)  
        &:= \frac{D_{\CC}(w, t)}{\binom{n}{w} (q - 1)^{w}}.
    \end{align*}
    Moreover, the probability that there exists at least one non-casual codeword within distance $t$ of the received word 
    is 
    \begin{align*}
        P_{\CC}(w, t) 
        &\le \ol_{\CC}(w, t) 
    \end{align*}
    for all $w, t$, with equality if $t < d / 2$. 
\end{theorem}

\begin{proof}
    This is essentially \cite[Theorem~D-1]{McEliece03} (see also \cite[Theorem~5.1]{McEliece03-1}, \cite{McEliece86} and \cite{Cheung89}), whose proof, although given for $[n, k]$ $\RS$ codes only, applies more generally to $\CC$ as above. 
\end{proof}

\hide{
Then \autoref{thm:CPlist} tells us that, to compute the average number of non-casual codewords of $\CC$ within distance $t$ of the received word, it is enough to know the numbers $D_{\CC}(w, t)$. 

\todo{Find good estimates/upper bounds for $D_{\CC}(w, t)$.}
}

\subsection{Bounds on 
Average List Size and 
Error Probability}

Using the RS subcode view of $\CCP$ mentioned at the beginning of \autoref{sec:decoding}, 
any bounds obtained for $\RS$ using \autoref{thm:CPlist} automatically apply to $\CCP$ (e.g. see \eqref{eq:bound1} below). 
However, any $f \in \fpkq$ has 
$k + 1$ coefficients, $1 + \lfloor k / p \rfloor$ of which are guaranteed to be zero. 
Therefore, 
we can expect a factor of $q^{1 + \lfloor k / p \rfloor}$ improvement on the list size of $\CCP$ over $\RS$ on average, as discussed below. 


Let $\CC = \GRS(\fpkq')$ and $\CC' = \GRS(\CF(k - 1, q))$. 
Then $\CPlistdecode(m', s)$ and 
$\GS(m', s)$ are 
bounded distance decoders
for $\CC$ and $\CC'$, respectively,
with decoding radius 
\begin{align*}
    t_{s} 
    &:= n - \tau_{s} \\ 
    &= n - 1 - \left \lfloor \frac{\lfloor \sqrt{(k - 1) n (s + 1)} \rfloor}{s} \right \rfloor \\ 
    &= n - 1 - \left \lfloor \sqrt{(k - 1) n \left(1 + \frac{1}{s}\right)} \right \rfloor,
\end{align*}
where (and henceforth) $n = q - 1$. Note that $t_{1}, t_{2}, \dots$ is a bounded non-decreasing sequence of positive integers. Hence,
there exists a positive integer $s_{0}$ such that 
\begin{align*}
    t_{s_{0}} = t_{s_{0} + 1} = \cdots = t_{\infty} := n - 1 - \lfloor \sqrt{(k - 1) n} \rfloor.
\end{align*}
Since $\CC \subseteq \CC'$, we have $B_{\CC}(u, t) \subseteq B_{\CC'}(u, t)$. Consequently, $D_{\CC}(w, t) \le D_{\CC'}(w, t)$ and $\ol_{\CC}(w, t) \le \ol_{\CC'}(w, t)$. 
Moreover, since $|B_{\CC}(w, t)|$ is non-decreasing in $t$, so are $D_{\CC}(w, t)$ and $\ol_{\CC}(w, t)$. Hence, by \autoref{thm:CPlist}, 
\begin{align}
    \label{eq:bound1}
    P_{\CC}(w, t_{s}) 
    &\le \ol_{\CC}(w, t_{s}) 
    \le \ol_{\CC'}(w, t_{s}),
\end{align}
for all $s \ge 1$.

\hide{
Observe that the average size of 
$B_{\CC}(u, t)$ over all possible received vectors $u \in \fq^{n}$ 
is simply equal to $|\CC| / |\CC'|$ times the average size of $B_{\CC'}(u, t)$, i.e. 
\begin{align*}
    \overline{|B_{\CC}(u, t)|} 
    &= \overline{|B_{\CC'}(u, t)|} \cdot \frac{|\CC'|}{|\CC|} 
    = \overline{|B_{\CC'}(u, t)|} \cdot q^{- \lfloor k / p \rfloor},
\end{align*} 
where the last equality follows from \eqref{eq:CPsize}.
Hence, on average, 
\begin{align*}
    D_{\CC}(w, t) 
    \le D_{\CC'}(w, t) \cdot q^{- \lfloor k / p \rfloor}. 
\end{align*}
}

\subsection{Average List Size Improvement}

Assuming that the coefficients of the polynomials 
corresponding to the codewords 
in the decoding sphere of $\GS(m', s)$ are random, only a $(1 / q)^{1 + \lfloor k / p \rfloor}$ fraction of them belong to $\fpkq$ on average. Hence, 
\begin{align*}
    \ol_{\CC}(w, t_{s}) 
    &= q^{-1 - \lfloor k / p \rfloor} 
    \ol_{\CC'}(w, t_{s}). 
\end{align*}
Then, by \autoref{thm:CPlist}, 
\begin{align}
    \label{eq:bound2}
    P_{\CC}(w, t) 
    &\le 
    q^{-1 - \lfloor k / p \rfloor} \ol_{\CC'}(w, t).
\end{align}
McEliece and Swanson~\cite{McEliece86} proved that for all $w$ and $t$, 
\begin{align*}
    \ol_{\CC'}(w, t) 
    &\le 
    \frac{1}{(q - 1)^{n - k}} 
    \sum_{i = d - w}^{t} 
    \binom{n}{i} (q - 1)^{i}.
\end{align*}
Writing $q - 1 = \theta / (1 - \theta)$ for $\theta = 1 - 1 / q \in (0, 1)$ gives 
\begin{align*}
    \ol_{\CC'}(w, t) 
    &\le 
    \frac{(1 - \theta)^{n - k}}{\theta^{n - k}} 
    \sum_{i = d - w}^{t} 
    \binom{n}{i} \frac{\theta^{i}}{(1 - \theta)^{i}} \\ 
    &= \frac{(1 - \theta)^{- k}}{\theta^{n - k}} 
    \sum_{i = d - w}^{t} 
    \binom{n}{i} \theta^{i} (1 - \theta)^{n - i} \\ 
    &= \frac{(1 - \theta)^{- k}}{\theta^{n - k}} 
    [F(t; n, \theta) - F(d - w - 1; n, \theta)] \\ 
    &\le \frac{(1 - \theta)^{- k}}{\theta^{n - k}} 
    F(t; n, \theta), 
\end{align*}
where 
$F$ is the cumulative distribution function of the binomial distribution. 
Hence, by \eqref{eq:bound2}, 
\begin{align}
    q^{1 + \lfloor k / p \rfloor} \ol_{\CC}(w, t_{s}) 
    &\le \frac{(1 - \theta)^{- k}}{\theta^{n - k}} 
    F(t_{s}; n, \theta) 
    \label{eq:mceliece1} \\ 
    &\le \frac{(1 - \theta)^{- k}}{\theta^{n - k}} 
    \exp\left( - n D\left(\frac{t_{s}}{n} \bigg\| \theta\right)\right) 
    \label{eq:chernoff} \\
    &= \frac{n^{k + t_{s}}}
    {t_{s}^{t_{s}} (n - t_{s})^{n - t_{s}}} 
    \le \frac{2^{n}}{n^{n - k - t_{s}}} 
    \label{eq:jensen}
\end{align}
for all $s \ge 1$. 
Here, 
\eqref{eq:chernoff} is by Chernoff's bound, and \eqref{eq:jensen} is by Jensen's inequality. Therefore, 
\begin{multline*}
    \ol_{\CC}(w, t_{\infty}) 
    \le \frac{2^{n}}{n^{\lfloor \sqrt{(k - 1) n} \rfloor - (k - 1)} (n + 1)^{1 + \lfloor k / p \rfloor}} 
    \\ 
    = \exp\left(
        n \log 2 - 
        \left(
            \lfloor \sqrt{(k - 1) n} \rfloor - (k - 1)
        \right) 
        \log n 
    \right. \\ 
    \left. 
        \vphantom{\left(\lfloor \sqrt{(k - 1) n} \rfloor\right)} 
        - \left(1 + \lfloor k / p \rfloor\right) \log(n + 1)
    \right).
\end{multline*}
Writing $k = R n$ and $q = p = n + 1$, the argument of the exponential is asymptotically 
\begin{align*}
    n 
    \left(
        \log 2 - (\sqrt{R} - R) \log n - \frac{(1 + R)}{n} 
        \log(n + 1)
    \right) 
    , 
\end{align*}
which tends to $- \infty$ as $n \to \infty$ for any fixed rate $R \in (0, 1)$. 
Therefore, by \eqref{eq:bound2}, the output of $\CPlistdecode$ is $q^{1 + \lfloor k / p \rfloor}$ times less likely to contain more than one codeword than that of $\GS$ for a GRS code of the same length and degree. 
\hide{
Writing $k = R n$ for $0 < R < 1$, the argument of the exponential is approximately 
\begin{align*}
& n \log 2 - \left(
        (\sqrt{R} n - R n) \log n - \frac{R n}{p} \log(1 + n)
    \right) \\ 
    &= n 
    \left(
        \log 2 - (\sqrt{R} - R) \log n + \frac{R}{p} \log(1 + n)
    \right) 
\end{align*}
which tends to $- \infty$ as $n \to \infty$ for any given $R \in (0, 1)$. By \eqref{eq:avg_L}, this implies that the list size is indeed one in most cases. 

\subsubsection*{Example}

Consider $\CC = \RS(\CF_{2}(15, 32))$ with $s = 3$.  
Then 
\begin{align*}
    \left\lfloor \frac{d}{2} \right\rfloor = 8, 
    \quad t_{s} = 9, 
    \quad t_{\infty} = 10
\end{align*}
so that 
\begin{align*}
    32^{15 / 2} P_{\CC}(w, t_{s}) 
    &\le 32^{15 / 2} \ol_{\CC}(w, t_{s}) 
    \le \frac{2^{31}}{31^{31 - 15 - 9}}
    \approx 0.078 
\end{align*}
and 
\begin{align*}
    32^{15 / 2} P_{\CC}(w, t_{\infty}) 
    &\le 32^{15 / 2} \ol_{\CC}(w, t_{\infty}) 
    \le \frac{2^{31}}{31^{31 - 15 - 10}}
    \approx 2.42, 
\end{align*}
which shows that the output of $\CPlistdecode$ almost always contains a single codeword, while the output of $\GS$ for 
$\CC'$ may contain multiple codewords.
(Note that we have only considered bounds derived from \eqref{eq:jensen} here for illustration purposes, while \eqref{eq:mceliece1} and \eqref{eq:chernoff} do give tighter bounds.) 
}

\subsubsection*{Simulation Results}

For each prime $q \in (5, 50)$, $5$ randomly chosen $f \in \fpkq'$ for $k = \lfloor q / 2 \rfloor - 1$ were encoded in $\CC$ and $\CC'$, and 
transmitted $200$ times with 
$t_{\infty}$ random errors
. 
The average list sizes of the corresponding list decoders are shown in \autoref{tab:list_size}. 
\hide{
A case where CP list size is 2 and RS list size is 4:
(18, 10, 9) CP code with character chi_1 associated with [18, 10, 9] Generalized Reed-Solomon Code over GF(19)
RS polynomial: 11*x^9 + 6*x^8 + 10*x^7 + 7*x^6 + 12*x^5 + 14*x^4 + 2*x^3 + 8*x^2 + 5*x + 9
CP polynomial: 11*x^9 + 6*x^8 + 10*x^7 + 7*x^6 + 12*x^5 + 14*x^4 + 2*x^3 + 8*x^2 + 5*x
CP codeword: (0.945817241 - 0.324699470*I, -0.0825793455 + 0.996584493*I, -0.986361303 - 0.164594590*I, -0.0825793455 + 0.996584493*I, 0.945817242 + 0.324699469*I, 1.00000000, 0.245485487 - 0.969400266*I, 0.546948158 - 0.837166478*I, 0.546948158 - 0.837166478*I, 0.945817242 + 0.324699469*I, -0.879473751 + 0.475947393*I, -0.986361303 + 0.164594591*I, -0.879473751 + 0.475947393*I, -0.0825793458 - 0.996584493*I, -0.0825793458 - 0.996584493*I, -0.677281572 + 0.735723911*I, -0.986361303 - 0.164594590*I, -0.0825793458 - 0.996584493*I)
RS codeword: (18, 10, 11, 1, 5, 0, 10, 14, 11, 10, 12, 13, 9, 6, 1, 17, 18, 5)
CP received word: (0.945817241 - 0.324699470*I, -0.0825793455 + 0.996584493*I, -0.986361303 - 0.164594590*I, -0.0825793455 + 0.996584493*I, 0.241483579 - 0.843659280*I, 1.00000000, 0.235460827 - 0.259279807*I, 0.546948158 - 0.837166478*I, 0.546948158 - 0.837166478*I, 0.0982023708 + 0.284480046*I, -0.951522083 + 0.723812145*I, -0.986361303 + 0.164594591*I, -0.879473751 + 0.475947393*I, -0.630942105 + 0.229026079*I, -0.0825793458 - 0.996584493*I, -0.677281572 + 0.735723911*I, -0.986361303 - 0.164594590*I, -0.0825793458 - 0.996584493*I)
RS received word: (18, 10, 11, 1, 13, 0, 10, 0, 11, 10, 6, 13, 5, 6, 6, 17, 18, 5)
CP list: [11*x^9 + 6*x^8 + 10*x^7 + 7*x^6 + 12*x^5 + 14*x^4 + 2*x^3 + 8*x^2 + 5*x, 9*x^9 + 3*x^7 + 9*x^6 + 18*x^5 + 7*x^4 + 14*x^3 + 8*x^2 + 4*x]
RS list: [11*x^9 + 6*x^8 + 10*x^7 + 7*x^6 + 12*x^5 + 14*x^4 + 2*x^3 + 8*x^2 + 5*x, 5*x^9 + 3*x^8 + 7*x^7 + x^6 + 8*x^5 + 13*x^4 + 10*x^3 + 3*x^2 + 10*x + 15, 4*x^9 + 15*x^8 + 11*x^7 + 10*x^6 + 9*x^5 + 12*x^4 + 2*x^3 + 7*x^2 + 14*x + 10, 16*x^8 + 15*x^7 + 10*x^6 + 2*x^5 + 3*x^4 + 2*x^3 + 18*x^2 + x + 6]
\begin{table}[!htbp]
    \caption{Monte Carlo Simulated Average List Sizes
    } 
    \centering
    \begin{tabular}{|c|c|c|}
        \hline 
        $(q, k)$ & $\CC$ & $\CC'$ \\ 
        \hline 
        $(8, 2)$ & $1.084$ & $3.344$ \\
        $(8, 3)$ & $1.081$ & $3.567$ \\ 
        $(9, 4)$ & $1.518$ & $5.092$ \\ 
        $(16, 3)$ & $1.03$ & $2.164$ \\ 
        $(16, 4)$ & $1$ & $2.326$ \\ 
        $(16, 5)$ & --- & --- \\ 
        \hline 
    \end{tabular}
    \label{tab:list_size}
\end{table}
\begin{table}[!htbp]
    \caption{Monte Carlo Simulated Average List Sizes
    } 
    \centering
    \begin{tabular}{|c|c|c|}
        \hline 
        $q$ & $\CC$ & $\CC'$ \\ 
        \hline 
        $7$ & $1$ & $1$ \\
        $11$ & $1$ & $1$ \\ 
        $13$ & $1$ & $1.598$ \\ 
        $17$ & $1$ & $1.136$ \\ 
        $19$ & $1.009$ & $1.725$ \\ 
        $23$ & $1.002$ & $1.123$ \\ 
        $29$ & $1.001$ & $1.089$ \\ 
        $31$ & $1$ & $1.003$ \\ 
        $37$ & $1$ & $1.001$ \\ 
        $41$ & $1$ & $1$ \\ 
        $43$ & $1$ & $1$ \\ 
        $47$ & $1$ & $1$ \\ 
        \hline 
    \end{tabular}
    \label{tab:list_size}
\end{table}
}
\begin{table}[!htbp]
    \caption{Monte Carlo Simulated Average List Sizes
    } 
    \centering
    \begin{tabular}{|c|c|c||c|c|c|}
        \hline 
        $q$ & $\CC$ & $\CC'$ & 
        $q$ & $\CC$ & $\CC'$ \\ 
        \hline 
        $7$ & $1$ & $1$ & 
        $29$ & $1.001$ & $1.072$ \\ 
        $11$ & $1$ & $1$ & 
        $31$ & $1$ & $1.002$ \\ 
        $13$ & $1$ & $1.585$ & 
        $37$ & $1$ & $1.004$ \\ 
        $17$ & $1$ & $1.121$ & 
        $41$ & $1$ & $1$ \\ 
        $19$ & $1.005$ & $1.711$ & 
        $43$ & $1$ & $1$ \\ 
        $23$ & $1.004$ & $1.115$ & 
        $47$ & $1$ & $1$ \\ 
        \hline 
    \end{tabular}
    \label{tab:list_size}
\end{table}

\section{Conclusion and Future Directions}
\label{sec:five}

In this work, we studied and analyzed efficient minimum-distance and list decoding for one-dimensional CP subspace codes over prime fields. 
Possible directions for future work include adapting our decoding algorithms in \autoref{sec:decoding} to CP codes over arbitrary finite fields and to higher-dimensional CP codes~\cite{Hessam21}. Note that the problem becomes inherently more challenging in the latter case as higher-dimensional CP codes cannot be directly mapped to block codes. Hence, studying this problem may require developing entirely new techniques. 


Note also that we only considered hard-decision decoding as the first step in decoding CP codes. Thus, another direction is to look at soft-decision decoding (see, e.g.~\cite{Koetter03}) of CP codes, when soft information is available from the mapping of the message coordinates to the finite field elements (i.e. via \eqref{eq:chi-inverse}). 

\bibliographystyle{IEEEtran}
\bibliography{IEEEabrv, references}

\clearpage

\hide{
\newpage

\appendix

\subsection{Guruswami--Sudan Algorithm for \texorpdfstring{$\GRS(\CF(k - 1, q))$}{GRS(F(k - 1, q))}}
\label{sec:gs}

Let $m \in \BC^{n}$ be the original message and $m' = m + e$ the received message. Then there is a unique $f \in \CF(k - 1, q)$ with $m_{i} = v_{i} f(\alpha_{i})$ for $i = 1, \dots, n$. Let $S = \{(\alpha_{i}, y_{i})\}_{i = 1}^{n}$, where $y_{i} = m_{i}' / v_{i}$. Note that if $m' = m$, then $y_{i} = f(\alpha_{i})$ for all $i$, i.e. all points of $S$ lie on the curve $Q(X, Y) = Y - f(X)$. 
The idea then is to find a polynomial $Q(X, Y) = \sum_{i, j} q_{i, j} X^{i} Y^{j} \in \fqxy$ that has a root at each point of $S$ with a certain multiplicity $s$ (a parameter called \emph{interpolation multiplicity}), and then find the factors of $Q(X, Y)$ of the form $Y - f(X)$. To formally state the algorithm, we require the following definitions.

\begin{definition}
    \label{def:wdeg}
    The \emph{$(\alpha, \beta)$-(weighted) degree} of $Q \in \fqxy$ is 
    $\deg_{\alpha, \beta}(Q) 
    := \max\{\alpha i + \beta j: q_{i, j} \neq 0\}$. 
\end{definition}

\hide{
\begin{definition}
    $N_{k}(c)$ is the number of monomials in $\fqxy$ of $(1, k)$-degree $\le c$. 
\end{definition}
}

\begin{definition}
    \label{def:multiplicity}
    The \emph{multiplicity} of a root $(\alpha, \beta)$ of $Q \in \fqxy$ is the smallest degree of a monomial in $Q(X + \alpha, Y + \beta)$. 
\end{definition}

Then the Guruswami--Sudan algorithm can be stated as follows. 

\hide{
\subsubsection*{Guruswami--Sudan Algorithm $\GS(s)$}

\begin{enumerate}
    \item {\bf Interpolation:} Let $c = \lfloor \sqrt{(k - 1) n s (s + 1)} \rfloor$ and set $\tau_{s} := \lfloor c / s \rfloor + 1$. Construct a non-zero polynomial $Q \in \fqxy$ with $\deg_{1, k - 1}(Q) \le c$ (see \autoref{def:wdeg}) such that each element of $S$ is a root of $Q$ of multiplicity at least $s$ (see \autoref{def:multiplicity}). 
    \label{alg:gs:interpolation}
    \item {\bf Factorization:} Find all factors of $Q(X, Y)$ of the form $Y - f(X)$ such that $\deg(f) \le k - 1$ and $f(\alpha_{i}) = y_{i}$ for at least $\tau_{s}$ of the $i$'s. Return the list of all such $f$. 
    \label{alg:gs:factorization}
\end{enumerate}
}


\begin{algorithm}[!htbp]
    \caption{$\GS(m', s)$}
    \label{alg:gs}
    \begin{algorithmic}[1]
        \Require{Received word $m'\in \BC^{n}$ with possibly $\ge d / 2$ errors; interpolation multiplicity $s \in \BZ_{+}$}
        \Ensure{All $f \in \CF(k - 1, q)$ with $\dH(m', \GRS(f)) \le n - \tau_{s}$}
        \State (Interpolation) Construct a non-zero polynomial $Q \in \fqxy$ with $\deg_{1, k - 1}(Q) \le c$ 
        such that each element of $S$ is a root of $Q$ of multiplicity 
        at least $s$.
        \label{alg:gs:interpolation}
        \State (Factorization) Find all factors of $Q(X, Y)$ of the form $Y - f(X)$ such that $\deg(f) \le k - 1$ and $f(\alpha_{i}) = y_{i}$ for at least $\tau_{s}$ of the $i$'s.
        \label{alg:gs:factorization}
        \State \Return All $f$ given by Line~\ref{alg:gs:factorization}.
    \end{algorithmic}
\end{algorithm}

\subsection{Derivation of \eqref{eq:mceliece1} from \eqref{eq:mceliece}}
\label{sec:derivation}

Write $q - 1 = \theta / (1 - \theta)$ for $\theta = 1 - 1 / q \in (0, 1)$. Then 
\begin{align*}
    \ol_{\CC'}(w, t) 
    &\le 
    \frac{1}{(q - 1)^{n - k}} 
    \sum_{i = d - w}^{t} 
    \binom{n}{i} (q - 1)^{i} \\ 
    &= \frac{(1 - \theta)^{n - k}}{\theta^{n - k}} 
    \sum_{i = d - w}^{t} 
    \binom{n}{i} \frac{\theta^{i}}{(1 - \theta)^{i}} \\ 
    &= \frac{(1 - \theta)^{- k}}{\theta^{n - k}} 
    \sum_{i = d - w}^{t} 
    \binom{n}{i} \theta^{i} (1 - \theta)^{n - i} \\ 
    &= \frac{(1 - \theta)^{- k}}{\theta^{n - k}} 
    [F(t; n, \theta) - F(d - w - 1; n, \theta)] \\ 
    &\le \frac{(1 - \theta)^{- k}}{\theta^{n - k}} 
    F(t; n, \theta) 
\end{align*}
Taking $t = t_{s}$ gives \eqref{eq:mceliece1}. 
}

\end{document}